\documentclass[11pt,letterpaper,oneside,peerreview,onecolumn,
draftcls]{IEEEtran}%
\usepackage{amsmath}
\usepackage{amsfonts}
\usepackage{amssymb}
\usepackage{graphicx}%
\providecommand{\U}[1]{\protect\rule{.1in}{.1in}}
\newtheorem{theorem}{Theorem}

\newtheorem{corollary}{Corollary}

\begin{document}

\title{Selective Cooperative Relaying over Time-Varying Channels}
\author{Diomidis S. Michalopoulos, \IEEEmembership{Student Member,~IEEE,} Athanasios
S. Lioumpas, \IEEEmembership{Student Member,~IEEE,} George K. Karagiannidis,
\IEEEmembership{Senior Member,~IEEE} and Robert Schober,
\IEEEmembership{Senior Member,~IEEE} \thanks{D. S. Michalopoulos, A. S.
Lioumpas, and G. K. Karagiannidis are with the Wireless Communications Systems
Group (WCSG), Electrical and Computer Engineering Department, Aristotle
University of Thessaloniki, 54124 Thessaloniki, Greece (e-mail:\{dmixalo;
alioumpa; geokarag\}@auth.gr).} \thanks{R. Schober is with the Department of
Electrical and Computer Engineering, The University of British Columbia,
Vancouver, BC V6T 1Z4, Canada (e-mail: rschober@ece.ubc.ca)}}
\maketitle

\begin{abstract}
In selective cooperative relaying only a single relay out of the set of
available relays is activated, hence the available power and bandwidth
resources are efficiently utilized. However, implementing selective
cooperative relaying in time-varying channels may cause frequent relay
switchings that deteriorate the overall performance. In this paper, we study
the rate at which a relay switching occurs in selective cooperative relaying
applications in time-varying fading channels. In particular, we derive
closed-form expressions for the relay switching rate (measured in Hz) for
opportunistic relaying (OR) and distributed switch and stay combining (DSSC).
Additionally, expressions for the average relay activation time for both of
the considered schemes are also provided, reflecting the average time that a
selected relay remains active until a switching occurs. Numerical results
manifest that DSSC yields considerably lower relay switching rates than OR,
along with larger average relay activation times, rendering it a better
candidate for implementation of relay selection in fast fading environments.

\end{abstract}

\newpage

\section{Introduction}

Cooperative relaying has been recently proposed as a means of achieving the
beneficial effects of diversity in wireless communications systems, without
employing multiple antennas at neither the receiver nor the transmitter. Its
operation is based upon the concept of employing wireless relaying terminals
that assist the communication between a source and a destination terminal by
receiving the message sent from the source, and then processing it
appropriately and forwarding it to the destination. The relaying terminals may
be either fixed, infrastructure-based terminals placed at selected spots in
urban environments aiming at extending the coverage while avoiding the
infrastructure cost that the deployment of base stations entails, or mobile,
hand-held devices. In any of the above cases, apart from the apparent
robustness against small-scale fading, cooperative relaying offers resilience
against large attenuations due to path-loss, as well as shadowing. This,
together with the advantages that cooperative relaying offers in the various
levels of the open system interconnection (OSI) protocol stack, renders the
cooperative concept a strong candidate for utilization in future wireless
networks \cite{B:Fitzek_Coop}.

The most common cooperative relaying protocols were introduced in
\cite{J:Laneman}, where the term "cooperative diversity" was used so as to
emphasize the diversity advantages of relay employment. In the same work, an
outage analysis of the cooperative diversity concept was also conducted,
showing remarkable performance benefits as compared to the case without
relaying. Nonetheless, the analysis in \cite{J:Laneman} concerns the scenario
where a single relay is available for cooperation. In cases where multiple
relaying terminals are utilized, the error performance can be dramatically
improved if the multiple relay transmissions occur in orthogonal channels and
are combined by a maximal ratio combiner (MRC) at the destination. However,
orthogonal channel utilization results in a reduced overall spectral
efficiency. To this end, activating only a single relay out of the set of
available relays has been shown to be an effective means of achieving
cooperative diversity while limiting the negative effects of orthogonal relay
transmissions, offering thus a good tradeoff between error performance and
spectral efficiency.

Previous works on single relay selection in cooperative relaying scenarios
include \cite{C:Adve_Impr_AF}-\nocite{J:Bletsas_Cooper}\nocite{J:Rel_Sel_PA}%
\cite{C:Tan_Sp_Eff}, where the relay selection was based on a maximum
signal-to-noise-ratio (SNR) policy, thus attaining a diversity order equal to
the number of available relays, i.e., the same diversity order as for the case
where all the relays are activated, yet with considerably higher spectral
efficiency. Due to the somewhat opportunistic usage of the available
resources, the relay selection protocol based upon the maximum SNR rule is
termed opportunistic relaying (OR)\nocite{J:Bletsas_Cooper}. A simpler
alternative to OR, for the case of two available relaying terminals, is the
so-called distributed switch and stay combining (DSSC) protocol proposed in
\cite{J:DSSC2}. According to DSSC, a single relay remains active for as long
as the corresponding SNR is greater than a predetermined threshold value;
should this condition be violated, a relay switching occurs. The DSSC protocol
hence requires only a single end-to-end channel estimation for each
transmission, reducing thus the overall complexity while achieving the same
outage performance as OR, albeit inferior error performance \cite{J:DSSC2}. In
the sequel, we use the term "selective cooperative relaying" to refer to the
OR and DSSC protocols, i.e., to protocols where single relay selection takes place.

\subsection{Motivation}

Despite the above benefits of selective cooperative relaying, however, a major
issue that needs to be addressed is the rate at which a switching of the
active relaying terminal occurs in practical scenarios where the fading in
each of the links involved is time-varying. In fact, this rate reflects the
number of times per second the system has to switch from one relaying terminal
to another, and corresponds to a complexity measure regarding the
implementation of selective cooperative relaying in practice. More
specifically, frequent relay switchings may cause synchronization problems due
to the fact that the system needs to repeat the initialization process each
time the active relay changes, in order to re-adapt to the channel conditions
of the new branch. Apparently, such synchronization readjustment leads to
increased implementation complexity, as well as potential delays and outages
which may cause severe information loss with ultimate deteriorating effects on performance.

Particularly for the case of DSSC, where only one end-to-end branch is
estimated in each transmission period, relay switchings have a negative impact
on channel estimation along with synchronization. This is because each time a
relay switching occurs, a previously idle relay needs to be \textquotedblleft
awakened\textquotedblright.\ Hence, a new training sequence needs to be
initiated, which may not be long enough to provide accurate channel
estimation, resulting in detection errors in addition to those owing to weak
channel conditions. To the best of the authors' knowledge, the concept of
relay switchings in time-varying fading environments has not been addressed in
the literature.

\subsection{Contribution}

In this paper, we study the effect of the time-varying nature of fading
channels in selective cooperative relaying applications. In particular, we
provide closed-form expressions for the relay switching rate (measured in Hz)
of OR and DSSC, as a function of the average channel gains and the maximum
Doppler frequency of each of the source-relay and relay-destination links
involved. These expressions consider the case of independent but not
necessarily identically distributed (i.n.i.d.) Rayleigh fading channels, and
account for both AF and decode-and-forward (DF) relaying. Particularly for OR,
we derive the relay switching rate for arbitrary numbers of participating
relaying terminals when operating over independent and identically distributed
(i.i.d.) Rayleigh fading channels, showing that this rate is an increasing
function in the number of relays. In addition to the switching rate, we obtain
closed-form expressions for the average relay activation time, which is
defined as the average time interval that the selected relay remains activated
until the system switches to another relay. A set of numerical examples is
provided showing that DSSC results in a considerably lower relay switching
rate, as well as a larger average relay activation time than OR. These results
indicate that DSSC may be preferable in fast fading scenarios where the
channel gains change rapidly making OR difficult to implement due to the
frequent relay switchings, despite DSSC's inferiority in terms of error
performance \cite{J:DSSC2}.

\subsection{Outline}

The rest of this paper is organized as follows. The mode of operation of
OR\ and DSSC, together with some basic definitions are given in Section
\ref{BKGD}. Section \ref{OR} provides expressions for the relay switching rate
and the average relay activation time of OR, when operating over Rayleigh
fading channels, while the corresponding expressions for DSSC are presented in
Section \ref{DSSC}. Some numerical examples are given in Section \ref{NUM}.
Finally, Section \ref{CON} concludes the paper.

\section{\label{BKGD}Background}

We consider the cooperative relaying setup where a source terminal, $S$,
communicates with a destination terminal, $D$, with the aid of $L$ relaying
terminals which are denoted here by $R_{i}$, $i\in\left\{  1,...,L\right\}  $.
The relaying terminals may operate in either the DF or the AF mode. In the
former case, the relays fully decode the received signal and forward a
noise-free symbol to the destination, while in the latter case the relaying
terminals are used as simple analog repeaters which amplify the received
signal and forward it to the destination without demodulating it.
Additionally, the relays are assumed to be half-duplex, in the sense that they
cannot receive and transmit simultaneously; instead, the source-relay and
relay-destination transmissions are assumed to occur in time-orthogonal channels.

In this paper, we adopt the general notation, $a_{AB}$, to denote the channel
gain of the link between terminals $A$ and $B$, so that the channel gain of,
e.g., the $S$-$R_{i}$ link, is represented by $a_{SR_{i}}$. The fading in each
of the links involved is assumed to be Rayleigh distributed, with probability
density function (PDF) given by%
\begin{equation}
f_{a_{AB}}\left(  x\right)  =\frac{2x}{\Omega_{AB}}\exp\left(  -\frac{x^{2}%
}{\Omega_{AB}}\right)  \label{Rayl}%
\end{equation}
where $\Omega_{AB}$ represents the average squared channel gain of the $A$-$B$
link, e.g., $\Omega_{SR_{i}}=E\left[  a_{SR_{i}}^{2}\right]  $ with $E\left[
\cdot\right]  $ denoting expectation. We denote the maximum Doppler frequency
of the $A$-$B$ link by $\mathcal{F}_{AB}$, e.g., the maximum Doppler frequency
of the $S$-$R_{i}$ link is denoted by $\mathcal{F}_{SR_{i}}$. Moreover, all
terminals are assumed to transmit with identical power, denoted by $P_{T}$,
while the noise power in all of the links involved is identical and denoted by
$N_{0}$.

Throughout this work, two selective cooperative relaying protocols are
considered: The opportunistic relaying (OR) protocol presented in
\cite{J:Bletsas_Cooper}, and a variant of the two-relay distributed switch and
stay combining (DSSC) protocol proposed in \cite{J:DSSC2}, which is referred
to as DSSC-B here; this notation is adopted since DSSC-B varies from DSSC in
exactly the same way as SSC-B varies from SSC-A in \cite{J:Alouini_switch}.
The reasoning behind studying DSSC-B instead of the original DSSC protocol
presented in \cite{J:DSSC2} lies in the fact that DSSC-B yields less frequent
relay switchings than DSSC, similarly as SSC-B yields less frequent switchings
than SSC-A \cite{J:Alouini_switch}. The modes of operation of OR and DSSC-B
are given in detail in the ensuing subsection. It is worth mentioning that in
each case, only a single relay out of the set of available relays is activated
and denoted by $R_{b}$, resulting in a somewhat distributed version of
selection combining for OR; and a distributed version of SSC-B for DSSC-B. The
selection is performed at the destination terminal, which collects the channel
state information (CSI) of all the links involved. Then, after determining the
\textquotedblleft best\textquotedblright\ relay, the destination sends a
feedback message to the relays indicating the activation of the selected relay
and the deactivation of the previously selected relay. All other terminals
remain inactive until they receive a proper activation message from the
destination.
%The relay selection may be implemented either in a
%decentralized or in a centralized fashion, depending on whether direct
%communication among the relays is possible. In the former case, the relays
%decide in a distributed manner which one would forward the information
%received from the source, without the need for a central unit. In the latter
%case, a central unit (which could be, e.g., the destination terminal) collects
%the channel state information (CSI) of all the links involved, and sends
%appropiate feedback to the selected relay indicating its activation. All other
%terminals remain inactive until they receive proper activation feedaback by
%the central unit.

\subsection{Mode of Operation of Schemes Under Consideration}

\subsubsection{Opportunistic Relaying (OR)}

The OR protocol consists of selecting a single relay out of the set of $L$
available relays, particularly the relay with the highest of some
appropriately defined metric, which accounts for both the $S$-$R_{i}$ and
$R_{i}$-$D$ links and corresponds to performance measures of the $i$th
end-to-end path. Under the assumption of equal power transmitted by the
potential relaying terminals, such metrics may be the $\min$ equivalent
defined as%
\begin{equation}
a_{i}=\min\left(  a_{SR_{i}},a_{R_{i}D}\right)  \label{min1}%
\end{equation}
or the \textquotedblleft half harmonic mean\textquotedblright\ equivalent
defined as%
\begin{equation}
a_{i}=\frac{a_{SR_{i}}a_{R_{i}D}}{a_{SR_{i}}+a_{R_{i}D}}. \label{hm1}%
\end{equation}
The $\min$ equivalent metric accounts for determining the end-to-end path
based on the weakest intermediate link. It is appropriate for DF relaying
because it corresponds to an outage-equivalent of the end-to-end DF channel,
since the outage probability of DF relaying equals the cumulative distribution
function (CDF) of $a_{i}$ evaluated at the outage threshold SNR. The half
harmonic mean equivalent is appropriate for AF relaying since it corresponds
to a tight approximation of the end-to-end SNR of the $S$-$R_{i}$-$D$ link
\cite{J:Anghel_Kaveh}. It should be noted, however, that both criteria lead to
approximately the same results, since the $\min$ equivalent represents a tight
upper upper bound of the half harmonic mean equivalent, particularly when the
SNRs of the source-relay and relay-destination links are very different, e.g.,
$a_{SR_{i}}\gg a_{R_{i}D}$ \cite{J:Ikki_AF_Nak}, \cite{B:Mean_Ineq}. For this
reason, in the sequel we adopt the $\max$-$\min$ criterion for determining the
selected relay, $R_{b}$, so that the relay associated with the maximum
"bottleneck" of the source-relay and relay-destination links is selected,
i.e.,
\begin{equation}
b=\underset{i\in\left\{  1,...,L\right\}  }{\arg}\max\min\left(  a_{SR_{i}%
},a_{R_{i}D}\right)  . \label{b1}%
\end{equation}

\subsubsection{DSSC-B}

The DSSC protocol \cite{J:DSSC2} applies to the case where there are two
relays available for cooperation. Its simplicity over OR lies in the fact that
in each training period, only a single end-to-end channel has to be estimated.
This estimation is used so as to check whether the active branch is of
sufficient quality. The system thus switches from one relay to another only if
the equivalent SNR of the active branch lies below a predefined switching
threshold, $T$. Using the $\min$ equivalent metric defined above, a relay
switching occurs when $a_{b}<\sqrt{T/\Gamma}$, where $\Gamma=P_{T}/N_{0}$
denotes the ratio of the power transmitted by the source and each relay
divided by the noise power, i.e., the common SNR without fading.

Nevertheless, there may exist transmission periods where both the available
end-to-end channels are not strong enough; in such cases the system switches
continuously from one relay to another, though without reaching the desired
SNR level. Therefore, in order to avoid these excessive, as well as unavailing
switchings, throughout this paper we study a variant of the DSSC protocol
denoted here by DSSC-B. According to DSSC-B, a switching occurs whenever the
SNR of the active branch down-crosses the switching threshold, $T$; the system
then stays connected with the new branch, regardless of whether the SNR of the
new branch is greater or lower than $T$, until this SNR down-crosses $T$.
Mathematically speaking, denoting by $a_{i}^{j}$ the $\min$ equivalent of the
$i$th branch for a transmission period $j$, the active relay, $R_{b}^{j+1}$,
for the ensuing transmission period is determined by%
\begin{equation}%
\begin{array}
[c]{c}%
\text{If }R_{b}^{j}=R_{1}\text{, }R_{b}^{j+1}=~\left\{
\begin{array}
[c]{c}%
R_{1}\text{ if }\left[  a_{1}^{j-1}<\sqrt{T/\Gamma}\text{ or }a_{1}^{j}%
>\sqrt{T/\Gamma}\right] \\
~~R_{2}\text{ if }\left[  a_{1}^{j-1}>\sqrt{T/\Gamma}\text{ and }a_{1}%
^{j}<\sqrt{T/\Gamma}\right]
\end{array}
\right. \\
\text{If }R_{b}^{j}=R_{2}\text{, }R_{b}^{j+1}=\text{ }\left\{
\begin{array}
[c]{c}%
R_{2}\text{ if }\left[  a_{2}^{j-1}<\sqrt{T/\Gamma}\text{ or }a_{2}^{j}%
>\sqrt{T/\Gamma}\right] \\
~~R_{1}\text{ if }\left[  a_{2}^{j-1}>\sqrt{T/\Gamma}\text{ and }a_{2}%
^{j}<\sqrt{T/\Gamma}\right]
\end{array}
\right.
\end{array}
. \label{b2}%
\end{equation}

It is interesting to note that, under the $\min$ equivalent criterion, the
end-to-end path between source and destination in both OR and DSSC-B is
treated as a \textit{virtual} channel with channel gain $a_{i}$ $=\min\left(
a_{SR_{i}},a_{R_{i}D}\right)  $. The OR scheme can be thus regarded as a
virtual selection diversity scheme, where the instantaneous channel gain of
the $i$th input branch is $a_{i}$; the DSSC-B scheme can be interpreted as a
virtual SSC-B scheme with channel gains $a_{1}$ and $a_{2}$.

\subsection{Basic Definitions}

The rest of the paper focuses on deriving expressions, as well as providing
numerical examples, for the following performance measures for both OR and DSSC-B:

\begin{itemize}
\item The \textit{relay switching rate}, defined as the number of times per
second that a switching of the active relay takes place; that is, the
destination stops receiving from a certain relaying terminal and connects with
another one.

\item The \textit{average relaying activation time}, defined as the average
time duration that the selected relay remains activated, starting from the
time it receives an activation message until the system switches to another
relaying terminal.
\end{itemize}

\section{\label{OR}Relay Switching Rates and Average Relay Activation Time of
OR}

\subsection{Two Relays, i.n.i.d. Fading}

Let us first consider the two-relay OR scenario, where the fading in all the
intermediate links involved is i.n.i.d.

\begin{theorem}
The relay switching rate of OR with two available relays and i.n.i.d. fading
is given by%
\begin{align}
SR_{OR}  &  =\frac{\pi\sqrt{2\Omega_{1}\Omega_{2}}\left[  \Omega_{SR_{1}%
}\Omega_{SR_{2}}\sqrt{\Omega_{R_{1}D}\mathcal{F}_{R_{1}D}^{2}+\Omega_{R_{2}%
D}\mathcal{F}_{R_{2}D}^{2}}+\Omega_{R_{1}D}\Omega_{SR_{2}}\sqrt{\Omega
_{SR_{1}}\mathcal{F}_{SR_{1}}^{2}+\Omega_{R_{2}D}\mathcal{F}_{R_{2}D}^{2}%
}\right]  }{\left(  \Omega_{1}+\Omega_{2}\right)  ^{3/2}\left(  \Omega
_{SR_{1}}+\Omega_{R_{1}D}\right)  \left(  \Omega_{SR_{2}}+\Omega_{R_{2}%
D}\right)  }\label{swr3}\\
&  +\frac{\pi\sqrt{2\Omega_{1}\Omega_{2}}\left[  \Omega_{R_{2}D}\Omega
_{SR_{1}}\sqrt{\Omega_{SR_{2}}\mathcal{F}_{SR_{2}}^{2}+\Omega_{R_{1}%
D}\mathcal{F}_{R_{1}D}^{2}}+\Omega_{R_{1}D}\Omega_{R_{2}D}\sqrt{\Omega
_{SR_{1}}\mathcal{F}_{SR_{1}}^{2}+\Omega_{SR_{2}}\mathcal{F}_{SR_{2}}^{2}%
}\right]  }{\left(  \Omega_{1}+\Omega_{2}\right)  ^{3/2}\left(  \Omega
_{SR_{1}}+\Omega_{R_{1}D}\right)  \left(  \Omega_{SR_{2}}+\Omega_{R_{2}%
D}\right)  }\nonumber
\end{align}
where $\Omega_{i}=E\left[  a_{i}^{2}\right]  =\Omega_{SR_{i}}\Omega_{R_{i}%
D}/\left(  \Omega_{SR_{i}}+\Omega_{R_{i}D}\right)  $ represents the average
squared value of $a_{i}$, $i\in\left\{  1,2\right\}  $.
\end{theorem}

\begin{proof}
Let us consider the random processes%
\begin{equation}
a_{1}\left(  t\right)  =\min\left(  a_{SR_{1}}\left(  t\right)  ,a_{R_{1}%
D}\left(  t\right)  \right)  \label{a1}%
\end{equation}%
\begin{equation}
a_{2}\left(  t\right)  =\min\left(  a_{SR_{2}}\left(  t\right)  ,a_{R_{2}%
D}\left(  t\right)  \right)  \label{a2}%
\end{equation}
which correspond to the fading processes of the virtual end-to-end channels
$S$-$R_{1}$-$D$ and $S$-$R_{2}$-$D$, respectively. Additionally, let us define
the random process $Z\left(  t\right)  $ as%
\begin{equation}
Z\left(  t\right)  =a_{1}\left(  t\right)  -a_{2}\left(  t\right)  \label{z1}%
\end{equation}
so that the active relay in each time instance $t$ is determined by the signum
of $Z\left(  t\right)  $ in this time instance, i.e., $R_{1}$ is active at
time $t$ if $Z\left(  t\right)  >0$; $R_{2}$ is active at time $t$ if
$Z\left(  t\right)  <0$. Consequently, in order to derive the average relay
switching rate it suffices to evaluate the average number of times the process
$Z\left(  t\right)  $ crosses zero. This is equivalent to obtaining the level
crossing rate (LCR) of $Z\left(  t\right)  $, evaluated at zero. Then, the
positive-going LCR of $Z\left(  t\right)  $ corresponds to the average number
of times the system switches from $R_{2}$ to $R_{1}$, while the negative-going
LCR of $Z\left(  t\right)  $ accounts for the average number of times the
system switches from $R_{1}$ to $R_{2}$.

The relay switching rate equals the sum of positive-going and negative-going
zero-crossing rates of $Z\left(  t\right)  $, which can be expressed as
\cite{J:Rice_sine}%
\begin{equation}
SR_{OR}=\int_{-\infty}^{0}\left\vert \overset{\cdot}{z}\right\vert f\left(
0,\overset{\cdot}{z}\right)  d\overset{\cdot}{z}+\int_{0}^{\infty}%
\overset{\cdot}{z}f\left(  0,\overset{\cdot}{z}\right)  d\overset{\cdot}{z}
\label{swr1}%
\end{equation}
where $f\left(  z,\overset{\cdot}{z}\right)  $ denotes the joint PDF of
$Z\left(  t\right)  $ and the time-derivative of $Z\left(  t\right)  $,
$\overset{\cdot}{Z}\left(  t\right)  $. Because of the independence of the
fading process and its time derivative, in each of the intermediate links
involved, owing to the Rayleigh fading assumption, the processes $Z\left(
t\right)  $ and $\overset{\cdot}{Z}\left(  t\right)  $ are independent
\cite{J:Beaulieu_SWR}. Therefore, $f\left(  z,\overset{\cdot}{z}\right)  $ can
be expressed as $f\left(  z,\overset{\cdot}{z}\right)  =f_{Z}\left(  z\right)
f_{\overset{\cdot}{Z}}\left(  \overset{\cdot}{z}\right)  $, with $f_{Z}\left(
\cdot\right)  $ and $f_{\overset{\cdot}{Z}}\left(  \overset{\cdot}{z}\right)
$ denoting the PDFs of $Z\left(  t\right)  $ and $\overset{\cdot}{Z}\left(
t\right)  $, respectively. Equation (\ref{swr1}) thus yields
\begin{subequations}
\label{swr2}%
\begin{align}
SR_{OR}  &  =f_{Z}\left(  0\right)  \left[  \int_{-\infty}^{0}\left\vert
\overset{\cdot}{z}\right\vert f_{\overset{\cdot}{Z}}\left(  \overset{\cdot}%
{z}\right)  d\overset{\cdot}{z}+\int_{0}^{\infty}\overset{\cdot}{z}%
f_{\overset{\cdot}{Z}}\left(  \overset{\cdot}{z}\right)  d\overset{\cdot}%
{z}\right] \label{swr2a}\\
&  =2f_{Z}\left(  0\right)  \int_{0}^{\infty}\overset{\cdot}{z}f_{\overset
{\cdot}{Z}}\left(  \overset{\cdot}{z}\right)  d\overset{\cdot}{z}
\label{swr2b}%
\end{align}
where we used the fact that the two integrals in (\ref{swr2a}) are equal to
each other since, apparently, the number of times the system switches from
$R_{1}$ to $R_{2}$ equals that of switching from $R_{2}$ to $R_{1}$, in the
long run. The PDF of $Z\left(  t\right)  $ evaluated at the origin is derived
as (see Appendix A)%
\end{subequations}
\begin{equation}
f_{Z}\left(  0\right)  =\frac{\sqrt{\pi}\sqrt{\Omega_{1}\Omega_{2}}}{\left(
\Omega_{1}+\Omega_{2}\right)  ^{3/2}}. \label{fz1}%
\end{equation}
Moreover, the second term in (\ref{swr2}) is derived as (see Appendix B)
\begin{align}
\int_{0}^{\infty}\overset{\cdot}{z}f_{\overset{\cdot}{Z}}\left(
\overset{\cdot}{z}\right)  d\overset{\cdot}{z}  &  =\frac{\sqrt{\pi}\left[
\Omega_{SR_{1}}\Omega_{SR_{2}}\sqrt{\Omega_{R_{1}D}\mathcal{F}_{R_{1}D}%
^{2}+\Omega_{R_{2}D}\mathcal{F}_{R_{2}D}^{2}}+\Omega_{R_{1}D}\Omega_{SR_{2}%
}\sqrt{\Omega_{SR_{1}}\mathcal{F}_{SR_{1}}^{2}+\Omega_{R_{2}D}\mathcal{F}%
_{R_{2}D}^{2}}\right]  }{\sqrt{2}\left(  \Omega_{SR_{1}}+\Omega_{R_{1}%
D}\right)  \left(  \Omega_{SR_{2}}+\Omega_{R_{2}D}\right)  }\nonumber\\
&  +\frac{\sqrt{\pi}\left[  \Omega_{R_{2}D}\Omega_{SR_{1}}\sqrt{\Omega
_{SR_{2}}\mathcal{F}_{SR_{2}}^{2}+\Omega_{R_{1}D}\mathcal{F}_{R_{1}D}^{2}%
}+\Omega_{R_{1}D}\Omega_{R_{2}D}\sqrt{\Omega_{SR_{1}}\mathcal{F}_{SR_{1}}%
^{2}+\Omega_{SR_{2}}\mathcal{F}_{SR_{2}}^{2}}\right]  }{\sqrt{2}\left(
\Omega_{SR_{1}}+\Omega_{R_{1}D}\right)  \left(  \Omega_{SR_{2}}+\Omega
_{R_{2}D}\right)  }. \label{I2}%
\end{align}
The relay switching rate of OR for the case of two relays with i.n.i.d. fading
channels is obtained by substituting (\ref{fz1}) and (\ref{I2}) into
(\ref{swr2b}), completing the proof.
\end{proof}

\begin{corollary}
\label{AATC1}The average relay activation time for OR with two available
relays and i.n.i.d. fading is given by%
\begin{equation}
AT_{i,OR}=\frac{2\Omega_{i}}{SR_{OR}\left(  \Omega_{1}+\Omega_{2}\right)
}\text{, \ }i\in\left\{  1,2\right\}  \label{ARAT1}%
\end{equation}
where $SR_{OR}$ is given in (\ref{swr3}).
\end{corollary}

\begin{proof}
Let us denote by $\rho_{i}^{OR}$ the steady-state probability of selecting
relay $R_{i}$ in the OR setup, i.e., $\rho_{1}^{OR}=\Pr\left\{  a_{1}%
>a_{2}\right\}  $; $\rho_{2}^{OR}=\Pr\left\{  a_{1}<a_{2}\right\}  $. Then,
given that the average switching rate from $R_{1}$ to $R_{2}$ equals the
average switching rate from $R_{2}$ to $R_{1}$, the average relay activation
time is derived as
\begin{equation}
AT_{i,OR}=\frac{2\rho_{i}^{OR}}{SR_{OR}}. \label{ARAT2}%
\end{equation}
Using the fact that $\rho_{i}^{OR}=\Omega_{i}/\left(  \Omega_{1}+\Omega
_{2}\right)  $ (see eqs. (\ref{ss1}), (\ref{ss2})), the proof is complete.
\end{proof}

\subsection{Two Relays, i.i.d. fading}

\begin{corollary}
The relay switching rate of OR with two available relays and i.i.d. fading is
given by%
\begin{equation}
SR_{OR}^{iid}=\frac{\pi\left[  \sqrt{\mathcal{F}_{SR_{1}}^{2}+\mathcal{F}%
_{SR_{2}}^{2}}+\sqrt{\mathcal{F}_{SR_{1}}^{2}+\mathcal{F}_{R_{2}D}^{2}}%
+\sqrt{\mathcal{F}_{R_{1}D}^{2}+\mathcal{F}_{SR_{2}}^{2}}+\sqrt{\mathcal{F}%
_{R_{1}D}^{2}+\mathcal{F}_{R_{2}D}^{2}}\right]  }{4\sqrt{2}}. \label{swr4}%
\end{equation}
Particularly for the case where the maximum Doppler frequencies are also
identical (i.e., $\mathcal{F}_{SR_{i}}=\mathcal{F}_{R_{i}D}=\mathcal{F}$,
$i\in\left\{  1,2\right\}  $), the relay switching rate is given by%
\begin{equation}
SR_{OR}^{iid}=\pi\mathcal{F}\text{.} \label{swr5}%
\end{equation}

\end{corollary}

\begin{proof}
The proof follows directly from (\ref{swr3}) after simple algebraic manipulations.
\end{proof}

One may note that the relay switching rate in the i.i.d. case is independent
of the channel amplitude and is determined only by the maximum Doppler
frequency in each of the links involved. Moreover, it is interesting to note
that (\ref{swr5}) yields a relay switching rate which is identical to that of
conventional selection diversity (i.e., where no relaying takes place) with
identical Rayleigh fading, as given in \cite{J:Beaulieu_SWR}, except for a
factor of $\sqrt{2}$. This implies that the OR setup can be considered as a
distributed selection diversity scheme with two \textit{virtual} Rayleigh
channels, where the maximum Doppler frequency of each virtual channel equals
$\sqrt{2}$-times the maximum Doppler frequency of the intermediate links. This
is due to the fact that in our case we deal with four time-varying links,
whereas in conventional selection diversity there are only two time-varying
links involved.

\begin{corollary}
The average relay activation time of two-relay OR for the i.i.d. scenario with
identical maximum Doppler frequencies is given by%
\begin{equation}
AT_{OR}^{iid}=\frac{1}{\pi\mathcal{F}}\text{.} \label{AAT1}%
\end{equation}

\end{corollary}

\begin{proof}
The proof follows from (\ref{ARAT1}), in conjunction with (\ref{swr5}).
\end{proof}

\subsection{$L$ Relays, i.i.d. Fading}

Let us now consider the versatile case of $L$ available relaying terminals.
For simplicity of the analysis, it is assumed that all the $S$-$R_{i}$ and
$R_{i}$-$D$ links experience i.i.d. fading, as well as identical maximum
Doppler frequency \footnote{We note that the results can be extended to the
case of i.n.i.d. fading and identical maximum Doppler frequency. However, the
resulting expressions would be too complicated and are out of the scope of
this paper.}.

\begin{theorem}
The relay switching rate of OR with $L$ available relays, i.i.d. Rayleigh
fading and identical maximum Doppler frequency $\mathcal{F}$ in each of the
intermediate channels involved, is given by%
\begin{equation}
SR_{OR}^{iid}=\sqrt{2}L\left(  L-1\right)  \pi\mathcal{F}\sum_{l=0}%
^{L-2}\left(  -1\right)  ^{l}\binom{L-2}{l}\left(  \frac{1}{l+2}\right)
^{\frac{3}{2}} \label{SRL2}%
\end{equation}

\end{theorem}

\begin{proof}
Similar to the case of two available end-to-end branches, the switching rate
for the $L$-branch case is evaluated through the use of the process%
\begin{equation}
Z\left(  t\right)  =a_{i}\left(  t\right)  -a_{k}\left(  t\right)
\end{equation}
where $a_{i}\left(  t\right)  $ is the fading process of any of the virtual
end-to-end channels as defined in (\ref{a1}) and%
\begin{equation}
a_{k}\left(  t\right)  =\max_{\substack{j\in\left\{  1,...,L\right\}
\\j\not =i}}a_{j}\left(  t\right)  \text{.}%
\end{equation}
It is important to note that, due to symmetry, the rate at which the system
switches from $R_{i}$ to any other relay is not affected by the index $i$ and
equals half of the overall relay switching rate, since all relays are selected
with identical probability. Consequently, the statistics of $Z\left(
t\right)  $ are the same for each $i\in\left\{  1,...,L\right\}  $, hence the
relay switching rate for this case is obtained as%
\begin{equation}
SR_{OR}^{iid}=Lf_{Z}\left(  0\right)  \int_{0}^{\infty}\overset{\cdot}%
{z}f_{\overset{\cdot}{Z}}\left(  \overset{\cdot}{z}\right)  d\overset{\cdot
}{z}\text{.} \label{SRL}%
\end{equation}
Using trivial integrations and the expressions for $f_{Z}\left(  0\right)  $
and $f_{\overset{\cdot}{Z}}\left(  \overset{\cdot}{z}\right)  $ given in
Appendix C, in (\ref{fz0L2}) and (\ref{fzbL}), respectively, (\ref{SRL})
yields (\ref{SRL2}); the proof is thus complete. As a cross-check, one may
notice that for $L=2$, (\ref{SRL2}) reduces to (\ref{swr5}).
\end{proof}

It is interesting to note from (\ref{SRL2}) that the relay switching rate is
an increasing function of $L$, implying that, as expected, the larger the
number of relays the more frequent relay switchings occurs.

\begin{corollary}
The average relay activation time for the i.i.d. scenario with arbitrary
number of available relays is derived as%
\begin{equation}
AT_{OR}^{iid}=\frac{1}{SR_{OR}^{iid}}. \label{AAT2}%
\end{equation}

\end{corollary}

\begin{proof}
Using the same approach as in Corollary \ref{AATC1}, it follows that
\[
AT_{OR}^{iid}=\frac{L\rho_{i}^{OR}}{SR_{OR}^{iid}}=\frac{1}{SR_{OR}^{iid}}%
\]
where we used the fact that, due to symmetry, $\rho_{i}^{OR}=1/L$.
\end{proof}

\section{\label{DSSC}Relay Switching Rates and Average Relay Activation Time
of DSSC-B}

\begin{theorem}
The relay switching rate for the two-relay DSSC-B setup is derived as%
\begin{align}
&  SR_{DSSC}\left(  T\right)  =\frac{\sqrt{2\pi}e^{-\frac{T\left(  \Omega
_{1}+\Omega_{2}\right)  }{\Gamma\Omega_{1}\Omega_{2}}}\left(  e^{\frac
{2T}{\Gamma\Omega_{2}}}-e^{\frac{T}{\Gamma\Omega_{2}}}\right)  \left(
\sqrt{\frac{T}{\Gamma\Omega_{SR_{1}}}}\mathcal{F}_{SR_{1}}+\sqrt{\frac
{T}{\Gamma\Omega_{R_{1}D}}}\mathcal{F}_{R_{1}D}\right)  }{2e^{\frac{T}%
{\Gamma\Omega_{1}}}-e^{\frac{2T}{\Gamma\Omega_{1}}}+2e^{\frac{T}{\Gamma
\Omega_{2}}}-e^{\frac{2T}{\Gamma\Omega_{2}}}-2e^{\frac{T\left(  \Omega
_{1}+\Omega_{2}\right)  }{\Gamma\Omega_{1}\Omega_{2}}}+e^{\frac{T\left(
\Omega_{1}+2\Omega_{2}\right)  }{\Gamma\Omega_{1}\Omega_{2}}}+e^{\frac
{T\left(  2\Omega_{1}+\Omega_{2}\right)  }{\Gamma\Omega_{1}\Omega_{2}}}%
-2}\nonumber\\
&  +\frac{\sqrt{2\pi}e^{-\frac{T\left(  \Omega_{1}+\Omega_{2}\right)  }%
{\Gamma\Omega_{1}\Omega_{2}}}\left(  e^{\frac{2T}{\Gamma\Omega_{1}}}%
-e^{\frac{T}{\Gamma\Omega_{1}}}\right)  \left(  \sqrt{\frac{T}{\Gamma
\Omega_{SR_{2}}}}\mathcal{F}_{SR_{2}}+\sqrt{\frac{T}{\Gamma\Omega_{R_{2}D}}%
}\mathcal{F}_{R_{2}D}\right)  }{2e^{\frac{T}{\Gamma\Omega_{1}}}-e^{\frac
{2T}{\Gamma\Omega_{1}}}+2e^{\frac{T}{\Gamma\Omega_{2}}}-e^{\frac{2T}%
{\Gamma\Omega_{2}}}-2e^{\frac{T\left(  \Omega_{1}+\Omega_{2}\right)  }%
{\Gamma\Omega_{1}\Omega_{2}}}+e^{\frac{T\left(  \Omega_{1}+2\Omega_{2}\right)
}{\Gamma\Omega_{1}\Omega_{2}}}+e^{\frac{T\left(  2\Omega_{1}+\Omega
_{2}\right)  }{\Gamma\Omega_{1}\Omega_{2}}}-2}\label{SRDSSC}\\
&  +\frac{\sqrt{2\pi}\left(  \sqrt{\frac{T}{\Gamma\Omega_{SR_{1}}}}%
\mathcal{F}_{SR_{1}}+\sqrt{\frac{T}{\Gamma\Omega_{R_{1}D}}}\mathcal{F}%
_{R_{1}D}+\sqrt{\frac{T}{\Gamma\Omega_{SR_{2}}}}\mathcal{F}_{SR_{2}}%
+\sqrt{\frac{T}{\Gamma\Omega_{R_{2}D}}}\mathcal{F}_{R_{2}D}\right)  \left(
e^{\frac{T}{\Gamma\Omega_{1}}}-1\right)  \left(  e^{\frac{T}{\Gamma\Omega_{2}%
}}-1\right)  }{2e^{\frac{T}{\Gamma\Omega_{1}}}-e^{\frac{2T}{\Gamma\Omega_{1}}%
}+2e^{\frac{T}{\Gamma\Omega_{2}}}-e^{\frac{2T}{\Gamma\Omega_{2}}}%
-2e^{\frac{T\left(  \Omega_{1}+\Omega_{2}\right)  }{\Gamma\Omega_{1}\Omega
_{2}}}+e^{\frac{T\left(  \Omega_{1}+2\Omega_{2}\right)  }{\Gamma\Omega
_{1}\Omega_{2}}}+e^{\frac{T\left(  2\Omega_{1}+\Omega_{2}\right)  }%
{\Gamma\Omega_{1}\Omega_{2}}}-2}\nonumber
\end{align}%
\[
\]

\end{theorem}

\begin{proof}
In the two-relay DSSC-B scheme, a relay switching occurs whenever the
negative-going slope of the output SNR crosses the switching threshold, $T$;
equivalently, a relay switching occurs whenever the process $a_{i}\left(
t\right)  $, $i\in\left\{  1,2\right\}  $, crosses $\sqrt{T/\Gamma}$ in a
negative-going direction; recall that $\Gamma$ denotes the common SNR without
fading. The switching rate for relay $R_{i}$, $i\in\left\{  1,2\right\}  $, is
given by%
\begin{equation}
N_{i}\left(  T\right)  =\int_{-\infty}^{0}\left\vert x\right\vert
f_{a_{i},\overset{\cdot}{a}_{i}}\left(  \sqrt{T/\Gamma},x\right)  dx
\label{NDSSC1}%
\end{equation}
where $f_{a_{i},\overset{\cdot}{a}_{i}}\left(  \cdot,\cdot\right)  $ denotes
the joint PDF of $a\left(  t\right)  $ and $\overset{\cdot}{a}_{i}\left(
t\right)  $. Using (\ref{fa1}), (\ref{fai}) and the fact that $a\left(
t\right)  $ and $\overset{\cdot}{a}_{i}\left(  t\right)  $ are independent
processes, (\ref{NDSSC1}) yields
\begin{equation}
N_{i}\left(  T\right)  =\sqrt{2\pi}\left(  \sqrt{\frac{T}{\Gamma\Omega
_{SR_{i}}}}\mathcal{F}_{SR_{i}}+\sqrt{\frac{T}{\Gamma\Omega_{R_{i}D}}%
}\mathcal{F}_{R_{i}D}\right)  \exp\left(  -\frac{\Omega_{SR_{i}}+\Omega
_{R_{i}D}}{\Omega_{SR_{i}}\Omega_{R_{i}D}}\frac{T}{\Gamma}\right)  \text{.}
\label{NDSSC2}%
\end{equation}

The steady-state probability of activating $R_{i}$ in the DSSC-B scheme is
derived in Appendix D as%
\begin{align}
\rho_{1}^{DSSC}\left(  T\right)   &  =\frac{F_{a_{2}^{2}}\left(
T/\Gamma\right)  \left[  F_{a_{2}^{2}}\left(  T/\Gamma\right)  -1\right]
\left[  1-F_{a_{1}^{2}}\left(  T/\Gamma\right)  +F_{a_{1}^{2}}^{2}\left(
T/\Gamma\right)  \right]  }{F_{a_{1}^{2}}\left(  T/\Gamma\right)  \left[
F_{a_{1}^{2}}\left(  T/\Gamma\right)  -1\right]  \left[  1-2F_{a_{2}^{2}%
}\left(  T/\Gamma\right)  +2F_{a_{2}^{2}}^{2}\left(  T/\Gamma\right)  \right]
+F_{a_{2}^{2}}\left(  T/\Gamma\right)  \left[  F_{a_{2}^{2}}\left(
T/\Gamma\right)  -1\right]  }\label{r1DSSC}\\
\rho_{2}^{DSSC}\left(  T\right)   &  =\frac{F_{a_{1}^{2}}\left(
T/\Gamma\right)  \left[  F_{a_{1}^{2}}\left(  T/\Gamma\right)  -1\right]
\left[  1-F_{a_{2}^{2}}\left(  T/\Gamma\right)  +F_{a_{2}^{2}}^{2}\left(
T/\Gamma\right)  \right]  }{F_{a_{2}^{2}}\left(  T/\Gamma\right)  \left[
F_{a_{2}^{2}}\left(  T/\Gamma\right)  -1\right]  \left[  1-2F_{a_{1}^{2}%
}\left(  T/\Gamma\right)  +2F_{a_{1}^{2}}^{2}\left(  T/\Gamma\right)  \right]
+F_{a_{1}^{2}}\left(  T/\Gamma\right)  \left[  F_{a_{1}^{2}}\left(
T/\Gamma\right)  -1\right]  } \label{r2DSSC}%
\end{align}
where
\begin{equation}
F_{a_{i}^{2}}\left(  x\right)  =1-\exp\left(  -x/\Omega_{i}\right)
\label{Fai2}%
\end{equation}
denotes the CDF of $a_{i}^{2}\left(  t\right)  $, $i\in\left\{  1,2\right\}
$. Therefore, the relay switching rate can be expressed as%
\begin{equation}
SR_{DSSC}\left(  T\right)  =\rho_{1}^{DSSC}\left(  T\right)  N_{1}\left(
T\right)  +\rho_{2}^{DSSC}\left(  T\right)  N_{2}\left(  T\right)  .
\label{SRDSSC01}%
\end{equation}
Substituting (\ref{NDSSC2})-(\ref{Fai2}) into (\ref{SRDSSC01}) and after some
algebraic manipulations, we arrive at (\ref{SRDSSC}).
\end{proof}

\begin{corollary}
The average relay activation time for DSSC-B is given by%
\begin{equation}
AT_{i,DSSC}\left(  T\right)  =\frac{2\rho_{i}^{DSSC}\left(  T\right)
}{SR_{DSSC}\left(  T\right)  } \label{ATDSSC}%
\end{equation}

\end{corollary}

\begin{proof}
The proof is identical to that of Corollary \ref{AATC1}. We note that the
factor $2$ in the numerator is present because the system switches to $R_{i}$
from the other available relay at a rate that equals $SR_{DSSC}\left(
T\right)  /2$ since, apparently, the number of times the system switches from
$R_{2}$ to $R_{1}$ equals the number of switches from $R_{1}$ to $R_{2}$, in
the long run.
\end{proof}

\section{\label{NUM}Numerical Examples and Discussion}

\subsection{Implementation Issues}

In classical diversity communication systems (e.g. conventional selection
combining or switch and stay combining, where no relaying takes place), the
switching between the diversity branches causes several problems, such as
\textquotedblleft an internal outage\textquotedblright\ due to the corruption
of the receiver filters and data signal chains, as well as phase estimation
failures \cite{J:Beaulieu_SWR}. In cooperative relaying systems, however,
where signals have to be exchanged between spatially distributed nodes, an
additional important issue has to be addressed; such issue is time
synchronization, which requires the local clocks of the relay nodes to be
synchronized, requiring various degrees of precision
\cite{J:Sivir_Sens_survey}.

It should be pointed out that in classical switched diversity systems, keeping
time synchronization between the transmitter and the receiver after switching
to a different diversity branch is not a serious problem, since the relative
time delays of the different branches are practically identical, because the
link distances are practically identical. On the contrary, in cooperative
relaying systems, especially in those employing mobile relays, the time delay
between the transmitter and the destination may change dramatically as soon as
a relay switching occurs. Time synchronization becomes particularly
challenging in cooperative networks since the network dynamics such as
propagation time or physical channel access time are in general
non-deterministic, because of the relative distances between source, relays,
and destination. In other words, switching to another relay would provoke a
time synchronization readjustment between the transmitting relay and the
destination, as well as the source. Therefore, it is evident that frequent
relay switchings are not desirable as they significantly increase the
implementation complexity of relay selection protocols, rendering them
difficult to implement in fast fading scenarios.

\subsection{Numerical Results}

In this subsection, we illustrate some numerical results, regarding the
relaying switching rate and average relaying activation time of both OR and
DSSC-B. All the results were also confirmed by simulations. As already
mentioned, the relay switching rate and the average relaying activation time
constitute a vital part in the design and implementation of selective
cooperative relaying, due to the complexity issues that each relay switching
entails. Fig. \ref{SR1} depicts the normalized switching rates\ of both OR and
DSSC-B versus the ratio $\mathcal{F}_{SR_{i}}/\mathcal{F}_{R_{i}D}$ for $L=2$
relays, when the fading powers in the $S$-$R$ and $R$-$D$\ links are
unbalanced. Specifically, the average channel gains of the $S$-$R_{1}$ and
$S$-$R_{2}$ links are assumed equal to each other, and the same is true for
the $R_{1}$-$D$ and $R_{2}$-$D$ links; the difference in the corresponding
average SNRs is 10dB, so that $\Gamma\Omega_{SR_{i}}=\Gamma\Omega_{R_{i}D}%
\pm10$dB, $i\in\left\{  1,2\right\}  $. As observed from Fig. \ref{SR1}, the
switching rate of DSSC-B is always smaller than that\ of OR, with the latter
to be under specific conditions even ten times greater than the former. This
verifies the intuition that DSSC-B leads to less frequent relay switchings,
accounting thus for simpler practical implementations. It should be noted that
in Fig. \ref{SR1} the switching threshold, $T$, used is that threshold which
leads to the maximum possible switching rate of DSSC-B, which is determined
through numerical optimization methods. In other words, the DSSC-B curves
depict the worst case in terms of switching rate, demonstrating that even in
this scenario the DSSC-B switching rate is much less than that of OR.

The effect of the number of available relays on the switching rate of OR for
the i.i.d. case with equal maximum Doppler frequencies in all of the links
involved, is plotted in Fig. \ref{SR2}. As expected, the switching rate
continuously increases as the number of the relays increases, yet in a
non-linear manner. The main result extracted from Fig. \ref{SR2} is that in
cases where the number of available relays is large, leaving some of them out
of the selection set may be preferable in practical applications with high
Doppler spread, despite the apparent performance cost.

The normalized average relay activation time of the relays for the same
conditions as in Fig. \ref{SR1} is plotted in Fig. \ref{AAT01}. Again the
average activation time in DSSC-B is the minimum possible, i.e., the parameter
$T$ is properly selected so as to account for the worst case. Fig. \ref{AAT01}
manifests that the average relay activation time of DSSC-B is considerably
larger than that of OR, regardless of the assumptions on the relative $S$-$R$
and $R$-$D$ channel strengths. The effect of the number of relays on the
activation time for OR\ is shown in Fig. \ref{AAT02}. One may observe that the
activation time reduces in a non-linear fashion as the number of relays\ reduces.

Finally, in Fig. \ref{SR3}, the switching rates of OR and DSSC-B\ are plotted
for $L=2$ relays and several values of the switching threshold, $T$, assuming
that the channel gains in the $S$-$R$ and $R$-$D$\ links are unbalanced, as
well as that $\mathcal{F}_{SR_{i}}=2\mathcal{F}_{R_{i}D}$. For each of the
scenarios shown here, the switching rate of DSSC-B\ is considerably lower than
that of OR. The ratio of switching rates is, however, crucially dependent on
the switching threshold of DSSC-B, implying that the difference of the
switching rates of OR and DSSC-B shown in Fig. \ref{SR1} is significantly
expanded for $T$ values different from the worst (in terms of relay switching
rate) case for DSSC-B.

\section{\label{CON}Conclusion}

We conducted a study of selective cooperative relaying in time-varying fading
channels. In particular, we derived the relay switching rate for selective
cooperative relaying, which reflects the concept of how frequently a
cooperative scheme switches from one relay to another. Together with the relay
switching rate, a closed-form expression for the average relay activation time
was also derived, which corresponds to the average time which a relay remains
activated until a switching occurs. Numerical results indicated that for the
case where there are two relays available, selecting the active relay in a max
SNR-based fashion (a.k.a. opportunistic relaying - OR) results in a
considerably higher switching rate than selecting the relay in a DSSC-B
fashion. Therefore, considering the complexity issues associated with frequent
relay switching, DSSC-B leads to a simpler implementation of selective
cooperative relaying than OR, and may be preferable in fast fading scenarios,
despite being inferior in terms of error performance.

\section*{Appendix A}

\section*{Evaluation of $f_{Z}\left(  0\right)  $}

The PDF of $Z\left(  t\right)  $ evaluated at zero reflects the probability
that the absolute difference of the virtual fading gains of the $S$-$R_{1}%
$-$D$ and $S$-$R_{2}$-$D$ channels lies in the infinitesimal interval $\left[
0,dz\right]  $. Therefore, $f_{Z}\left(  0\right)  $ can be expressed as%
\begin{equation}
f_{Z}\left(  0\right)  =\int_{0}^{\infty}f_{a_{1}}\left(  x\right)  f_{a_{2}%
}\left(  x\right)  dx \label{fz01}%
\end{equation}
where $f_{a_{1}}\left(  \cdot\right)  $ and $f_{a_{2}}\left(  \cdot\right)  $
denote the PDFs of the processes $a_{1}\left(  t\right)  $ and $a_{2}\left(
t\right)  $, respectively. Using (\ref{Rayl}) and (\ref{min1}), the CDF of
$a_{i}\left(  t\right)  $, $i\in\left\{  1,2\right\}  $, is expressed as%
\begin{align}
F_{a_{i}}\left(  x\right)   &  =1-\int_{x}^{\infty}f_{SR_{i}}\left(
\omega\right)  d\omega\int_{x}^{\infty}f_{R_{i}D}\left(  \omega\right)
d\omega\nonumber\\
&  =1-\exp\left(  -\frac{\Omega_{SR_{i}}+\Omega_{R_{i}D}}{\Omega_{SR_{i}%
}\Omega_{R_{i}D}}x^{2}\right)  \label{Fa1}%
\end{align}
while the PDF of $a_{i}\left(  t\right)  $ is derived by differentiating
(\ref{Fa1}), yielding%
\begin{equation}
f_{a_{i}}\left(  x\right)  =2x\frac{\Omega_{SR_{i}}+\Omega_{R_{i}D}}%
{\Omega_{SR_{i}}\Omega_{R_{i}D}}\exp\left(  -\frac{\Omega_{SR_{i}}%
+\Omega_{R_{i}D}}{\Omega_{SR_{i}}\Omega_{R_{i}D}}x^{2}\right)  . \label{fa1}%
\end{equation}
Note that the process $a_{i}\left(  t\right)  $ is also Rayleigh distributed;
its average squared value is denoted by $\Omega_{i}$ and is given by%
\begin{equation}
\Omega_{i}=\frac{\Omega_{SR_{i}}\Omega_{R_{i}D}}{\Omega_{SR_{i}}+\Omega
_{R_{i}D}}. \label{Wi}%
\end{equation}
Then, $f_{Z}\left(  0\right)  $ is derived as the integral of the product of
two Rayleigh distributions, yielding%
\begin{equation}
f_{Z}\left(  0\right)  =\frac{4}{\Omega_{1}\Omega_{2}}\int_{0}^{\infty}%
x^{2}\exp\left(  -\frac{\Omega_{1}+\Omega_{2}}{\Omega_{1}\Omega_{2}}%
x^{2}\right)  dx. \label{fz001}%
\end{equation}
Using \cite[eq. (3.321.5)]{B:Gra_Ryz_Book}, (\ref{fz001}) yields (\ref{fz1}).

\section*{Appendix B}

\section*{Derivation of $f_{\overset{\cdot}{Z}}\left(  \cdot\right)  $}

The time derivative of $Z\left(  t\right)  $, $\overset{\cdot}{Z}\left(
t\right)  $, equals the difference of the time derivatives $\overset{\cdot}%
{a}_{1}\left(  t\right)  $ and $\overset{\cdot}{a}_{2}\left(  t\right)  $.
Before proceeding in deriving $f_{\overset{\cdot}{Z}}\left(  \overset{\cdot
}{z}\right)  $, we first obtain the PDF of $\overset{\cdot}{a}_{i}\left(
t\right)  $, $i\in\left\{  1,2\right\}  $, as%
\begin{equation}
f_{\overset{\cdot}{a}_{i}}\left(  x\right)  =\Pr\left\{  a_{SR_{i}}\leq
a_{R_{i}D}\right\}  f_{\overset{\cdot}{a}_{SR_{i}}}\left(  x\right)
+\Pr\left\{  a_{SR_{i}}>a_{R_{i}D}\right\}  f_{\overset{\cdot}{a}_{R_{i}D}%
}\left(  x\right)  \label{fa2}%
\end{equation}
where $f_{\overset{\cdot}{a}_{SR_{i}}}\left(  \cdot\right)  $ and
$f_{\overset{\cdot}{a}_{R_{i}D}}\left(  \cdot\right)  $ denote the PDF of the
time derivatives of $a_{SR_{i}}$ and $a_{R_{i}D}$, respectively. Given that
$a_{SR_{i}}$ and $a_{R_{i}D}$ are Rayleigh distributed, $f_{\overset{\cdot}%
{a}_{SR_{i}}}\left(  \cdot\right)  $ and $f_{\overset{\cdot}{a}_{R_{i}D}%
}\left(  \cdot\right)  $ can be expressed as \cite{J:Rice_sine}%
\begin{equation}
f_{\overset{\cdot}{a}_{SR_{i}}}\left(  x\right)  =\frac{1}{\sqrt{2\pi}%
\overset{\cdot}{\sigma}_{SR_{i}}}\exp\left(  -\frac{x^{2}}{2\overset{\cdot
}{\sigma}_{SR_{i}}^{2}}\right)  \label{fsr1}%
\end{equation}%
\begin{equation}
f_{\overset{\cdot}{a}_{R_{i}D}}\left(  x\right)  =\frac{1}{\sqrt{2\pi}%
\overset{\cdot}{\sigma}_{R_{i}D}}\exp\left(  -\frac{x^{2}}{2\overset{\cdot
}{\sigma}_{R_{i}D}^{2}}\right)  \label{fr1d}%
\end{equation}
i.e., $\overset{\cdot}{a}_{SR_{i}}$ and $\overset{\cdot}{a}_{R_{i}D}$ are
zero-mean Gaussian random variables (RVs) with standard deviations
$\overset{\cdot}{\sigma}_{SR_{i}}=\pi\mathcal{F}_{SR_{i}}\sqrt{\Omega_{SR_{i}%
}}$ and $\overset{\cdot}{\sigma}_{R_{i}D}=\pi\mathcal{F}_{R_{i}D}\sqrt
{\Omega_{R_{i}D}}$, respectively. The steady-state probabilities $\Pr\left\{
a_{SR_{i}}\leq a_{R_{i}D}\right\}  $ and $\Pr\left\{  a_{SR_{i}}>a_{R_{i}%
D}\right\}  $ are given by%
\begin{equation}
\Pr\left\{  a_{SR_{i}}\leq a_{R_{i}D}\right\}  =\int_{0}^{\infty}f_{a_{SR_{i}%
}}\left(  x\right)  \left[  1-F_{a_{R_{i}D}}\left(  x\right)  \right]
dx=\frac{\Omega_{R_{i}D}}{\Omega_{SR_{i}}+\Omega_{R_{i}D}} \label{ss1}%
\end{equation}%
\begin{equation}
\Pr\left\{  a_{SR_{i}}>a_{R_{i}D}\right\}  =\frac{\Omega_{SR_{i}}}%
{\Omega_{SR_{i}}+\Omega_{R_{i}D}}. \label{ss2}%
\end{equation}
The PDF of $\overset{\cdot}{a}_{i}\left(  t\right)  $ is therefore derived as%
\begin{equation}
f_{\overset{\cdot}{a}_{i}}\left(  x\right)  =\frac{1}{\sqrt{2}\pi^{3/2}\left(
\Omega_{SR_{i}}+\Omega_{R_{i}D}\right)  }\left[  \frac{\Omega_{R_{i}D}%
\exp\left(  -\frac{x^{2}}{2\pi^{2}\mathcal{F}_{SR_{i}}^{2}\Omega_{SR_{i}}%
}\right)  }{\sqrt{\Omega_{SR_{i}}}\mathcal{F}_{SR_{i}}}+\frac{\Omega_{SR_{i}%
}\exp\left(  -\frac{x^{2}}{2\pi^{2}\mathcal{F}_{R_{i}D}^{2}\Omega_{R_{i}D}%
}\right)  }{\sqrt{\Omega_{R_{i}D}}\mathcal{F}_{R_{i}D}}\right]  \label{fai}%
\end{equation}
where $i\in\left\{  1,2\right\}  $. Consequently, it follows from (\ref{z1})
that the PDF of the process $\overset{\cdot}{Z}\left(  t\right)  $ is
expressed as \cite[ch. 6]{B:Papoulis_Book}%
\begin{equation}
f_{\overset{\cdot}{Z}}\left(  x\right)  =\int_{-\infty}^{\infty}%
f_{\overset{\cdot}{a}_{1}}\left(  y+x\right)  f_{\overset{\cdot}{a}_{2}%
}\left(  y\right)  dy. \label{fzb1}%
\end{equation}
Substituting (\ref{fai}) in (\ref{fzb1}) and making use of the integral
\cite[eq. (3.323.2)]{B:Gra_Ryz_Book}%
\begin{equation}
\int_{-\infty}^{\infty}\exp\left(  -\frac{\left(  y+x\right)  ^{2}}{\alpha
_{1}}\right)  \exp\left(  -\frac{x^{2}}{\alpha_{2}}\right)  dy=\sqrt{\pi
\frac{\alpha_{1}\alpha_{2}}{\alpha_{1}+\alpha_{2}}}\exp\left(  -\frac{x^{2}%
}{\alpha_{1}+\alpha_{2}}\right)  \label{I1}%
\end{equation}
we derive the PDF of $\overset{\cdot}{Z}\left(  t\right)  $, $f_{\overset
{\cdot}{Z}}\left(  \cdot\right)  $, as%
\begin{align}
f_{\overset{\cdot}{Z}}\left(  x\right)   &  =\frac{1}{\sqrt{2}\pi^{3/2}\left(
\Omega_{SR_{1}}+\Omega_{R_{1}D}\right)  \left(  \Omega_{SR_{2}}+\Omega
_{R_{2}D}\right)  }\left[  \frac{\Omega_{SR_{1}}\Omega_{SR_{2}}\exp\left(
-\frac{x^{2}}{2\pi^{2}\left(  \Omega_{R_{1}D}\mathcal{F}_{R_{1}D}^{2}%
+\Omega_{R_{2}D}\mathcal{F}_{R_{2}D}^{2}\right)  }\right)  }{\sqrt
{\Omega_{R_{1}D}\mathcal{F}_{R_{1}D}^{2}+\Omega_{R_{2}D}\mathcal{F}_{R_{2}%
D}^{2}}}\right. \nonumber\\
&  +\frac{\Omega_{R_{1}D}\Omega_{SR_{2}}\exp\left(  -\frac{x^{2}}{2\pi
^{2}\left(  \Omega_{SR_{1}}\mathcal{F}_{SR_{1}}^{2}+\Omega_{R_{2}D}%
\mathcal{F}_{R_{2}D}^{2}\right)  }\right)  }{\sqrt{\Omega_{SR_{1}}%
\mathcal{F}_{SR_{1}}^{2}+\Omega_{R_{2}D}\mathcal{F}_{R_{2}D}^{2}}}%
+\frac{\Omega_{R_{2}D}\Omega_{SR_{1}}\exp\left(  -\frac{x^{2}}{2\pi^{2}\left(
\Omega_{SR_{2}}\mathcal{F}_{SR_{2}}^{2}+\Omega_{R_{1}D}\mathcal{F}_{R_{1}%
D}^{2}\right)  }\right)  }{\sqrt{\Omega_{SR_{2}}\mathcal{F}_{SR_{2}}%
^{2}+\Omega_{R_{1}D}\mathcal{F}_{R_{1}D}^{2}}}\nonumber\\
&  +\left.  \frac{\Omega_{R_{1}D}\Omega_{R_{2}D}\exp\left(  -\frac{x^{2}}%
{2\pi^{2}\left(  \Omega_{SR_{1}}\mathcal{F}_{SR_{1}}^{2}+\Omega_{SR_{2}%
}\mathcal{F}_{SR_{2}}^{2}\right)  }\right)  }{\sqrt{\Omega_{SR_{1}}%
\mathcal{F}_{SR_{1}}^{2}+\Omega_{SR_{2}}\mathcal{F}_{SR_{2}}^{2}}}\right]  .
\label{fzb2}%
\end{align}

Having an expression for $f_{\overset{\cdot}{Z}}\left(  \cdot\right)  $, the
last term in (\ref{swr2}) is derived using \cite[eq. (3.321.4)]%
{B:Gra_Ryz_Book} as shown in (\ref{I2}).

\section*{Appendix C}

\section*{Relay Switching Rate of OR with $L$ available relays}

\subsubsection*{Evaluation of $f_{Z}\left(  0\right)  $}

Let us first derive the PDF of $a_{k}\left(  t\right)  $ as the PDF of the
maximum of $L-1$ i.i.d. Rayleigh RVs, yielding%
\begin{equation}
f_{a_{k}}\left(  x\right)  =\sum_{l=1}^{L-1}f_{a_{l}}\left(  x\right)
\prod_{\substack{j=1\\j\not =l}}^{L-1}F_{a_{j}}\left(  x\right)  =\left(
L-1\right)  f_{a}\left(  x\right)  \left[  F_{a}\left(  x\right)  \right]
^{L-2} \label{fak}%
\end{equation}
with $f_{a}\left(  x\right)  =\left(  2x/\Omega\right)  \exp\left(
-x^{2}/\Omega\right)  $ and $F_{a}\left(  x\right)  =1-\exp\left(
-x^{2}/\Omega\right)  $, where $\Omega=\Omega_{SR_{1}}/2=\Omega_{R_{1}%
D}/2=...=\Omega_{SR_{L}}/2=\Omega_{R_{L}D}/2$ denotes the average squared
channel gain in each of the virtual end-to-end channels involved. Then,
$f_{Z}\left(  0\right)  $ is derived as%
\begin{equation}
f_{Z}\left(  0\right)  =\int_{0}^{\infty}f_{a}\left(  x\right)  f_{a_{k}%
}\left(  x\right)  dx=\left(  L-1\right)  \int_{0}^{\infty}f_{a}^{2}\left(
x\right)  \left[  F_{a}\left(  x\right)  \right]  ^{L-2}dx\text{.}
\label{fz0L}%
\end{equation}
Using the product expansion%
\begin{equation}
\left[  1-\exp\left(  -\frac{x^{2}}{\Omega}\right)  \right]  ^{L-2}%
=1+\sum_{l=1}^{L-2}\left(  -1\right)  ^{l}\binom{L-2}{l}\exp\left(
-l\frac{x^{2}}{\Omega}\right)
\end{equation}
(\ref{fz0L}) yields%
\begin{equation}
f_{Z}\left(  0\right)  =\frac{\left(  L-1\right)  \sqrt{\pi}}{\Omega^{2}}%
\sum_{l=0}^{L-2}\left(  -1\right)  ^{l}\binom{L-2}{l}\left(  \frac{\Omega
}{l+2}\right)  ^{\frac{3}{2}}\text{.} \label{fz0L2}%
\end{equation}

\subsubsection*{Derivation of $f_{\overset{\cdot}{Z}}\left(  \cdot\right)  $}

Because of the i.i.d. assumption, the PDF of the time-derivative
$\overset{\cdot}{a}_{i}\left(  t\right)  $ is derived from (\ref{fai}) as%
\begin{equation}
f_{\overset{\cdot}{a}_{i}}\left(  x\right)  =\frac{1}{2\pi^{3/2}%
\mathcal{F}\sqrt{\Omega}}\exp\left(  -\frac{x^{2}}{4\pi^{2}\mathcal{F}%
^{2}\Omega}\right)
\end{equation}
that is, $\overset{\cdot}{a}_{i}$ is a zero-mean Gaussian RV with standard
deviation $\overset{\cdot}{\sigma}_{a_{i}}=\pi\mathcal{F}\sqrt{\Omega}$. The
PDF of the time-derivative of $a_{k}\left(  t\right)  $, $\overset{\cdot}%
{a}_{k}\left(  t\right)  $, is derived as%
\begin{equation}
f_{\overset{\cdot}{a}_{k}}\left(  x\right)  =\sum_{j=1}^{L-1}\rho_{j}%
^{OR}f_{\overset{\cdot}{a}_{j}}\left(  x\right)  =f_{\overset{\cdot}{a}_{i}%
}\left(  x\right)
\end{equation}
which implies that, due to symmetry, the PDF of the time-derivative of the
maximum of $L-1$ i.i.d. RVs, equals the PDF of each of the $L-1$ RVs. Using
(\ref{fzb1}), the PDF of $\overset{\cdot}{Z}\left(  t\right)  =\overset{\cdot
}{a}_{i}\left(  t\right)  -\overset{\cdot}{a}_{k}\left(  t\right)  $ is
derived as%
\begin{equation}
f_{\overset{\cdot}{Z}}\left(  x\right)  =\frac{1}{\left(  2\pi\right)
^{3/2}\mathcal{F}\sqrt{\Omega}}\exp\left(  -\frac{x^{2}}{8\pi^{2}%
\mathcal{F}^{2}\Omega}\right)  . \label{fzbL}%
\end{equation}

\section*{Appendix D}

\section*{Steady-State Relay Activation Probabilities of DSSC-B}

Considering that relay switchings in DSSC-B are determined in exactly the same
way as branch switchings in SSC-B, the steady-state relay activation
probabilites for DSSC-B are derived through the Markov states of SSC-B given
in \cite{J:Alouini_switch}. Specifically, the Markov chain of DSSC-B yields
six states, as follows. State 1 corresponds to the case where
\textquotedblleft$R_{1}$ is active and the overal SNR down-crosses
$T$\textquotedblright; state 2 corresponds to \textquotedblleft$R_{1}$ is
active and the overall SNR is below $T$\textquotedblright; state 3 corresponds
to \textquotedblleft$R_{1}$ is active and the overall SNR is greater than
$T$\textquotedblright; states 4, 5, 6 refer to the case where $R_{2}$ is
active, and are defined analogous to 1, 2, 3, respectively. The stationary
probabilities, $\pi_{j}^{DSSC-B},$ $j\in\left\{  1,...,6\right\}  $, of the
above Markov states are taken from \cite[eq. (21)]{J:Alouini_switch},
yielding
\begin{subequations}
\label{stat}%
\begin{align}
\pi_{1}^{DSSC-B}  &  =\frac{\left[  1-q_{1}\right]  q_{1}\left[
1-q_{2}\right]  q_{2}}{\left[  q_{1}+q_{2}\right]  \left[  1+2q_{1}%
q_{2}\right]  -\left[  q_{1}+q_{2}\right]  ^{2}-2q_{1}^{2}q_{2}^{2}}\\
\pi_{2}^{DSSC-B}  &  =\frac{q_{1}^{2}q_{2}\left[  1-q_{2}\right]  }{\left[
q_{1}+q_{2}\right]  \left[  1+2q_{1}q_{2}\right]  -\left[  q_{1}+q_{2}\right]
^{2}-2q_{1}^{2}q_{2}^{2}}\\
\pi_{3}^{DSSC-B}  &  =\frac{\left[  1-q_{1}\right]  ^{2}\left[  1-q_{2}%
\right]  q_{2}}{\left[  q_{1}+q_{2}\right]  \left[  1+2q_{1}q_{2}\right]
-\left[  q_{1}+q_{2}\right]  ^{2}-2q_{1}^{2}q_{2}^{2}}\\
\pi_{4}^{DSSC-B}  &  =\frac{\left[  1-q_{1}\right]  q_{1}\left[
1-q_{2}\right]  q_{2}^{2}}{\left[  q_{1}+q_{2}\right]  \left[  1+2q_{1}%
q_{2}\right]  -\left[  q_{1}+q_{2}\right]  ^{2}-2q_{1}^{2}q_{2}^{2}}\\
\pi_{5}^{DSSC-B}  &  =\frac{q_{2}^{2}q_{1}\left[  1-q_{1}\right]  }{\left[
q_{1}+q_{2}\right]  \left[  1+2q_{1}q_{2}\right]  -\left[  q_{1}+q_{2}\right]
^{2}-2q_{1}^{2}q_{2}^{2}}\\
\pi_{6}^{DSSC-B}  &  =\frac{\left[  1-q_{1}\right]  q_{1}\left[
1-q_{2}\right]  q_{2}}{\left[  q_{1}+q_{2}\right]  \left[  1+2q_{1}%
q_{2}\right]  -\left[  q_{1}+q_{2}\right]  ^{2}-2q_{1}^{2}q_{2}^{2}}%
\end{align}
where $q_{1}=F_{a_{1}^{2}}\left(  T/\Gamma\right)  $; $q_{2}=F_{a_{2}^{2}%
}\left(  T/\Gamma\right)  $. Considering that $R_{1}$ is active in the states
1, 2, 3, while $R_{2}$ in states 4, 5, 6, the steady-state relay activation
probabilities of DSSC-B are derived as $\rho_{1}^{DSSC}=\sum_{j=1}^{3}\pi
_{j}^{DSSC-B}$; $\rho_{2}^{DSSC}=\sum_{j=4}^{6}\pi_{j}^{DSSC-B}$, yielding
(\ref{r1DSSC}) and (\ref{r2DSSC}).\pagebreak

\bibliographystyle{IEEEtran}
\bibliography{acompat,References}

\bigskip\pagebreak%

\begin{figure}
[ptb]
\begin{center}
\includegraphics[
natheight=8.460500in,
natwidth=10.966700in,
height=4.5792in,
width=5.9283in
]%
{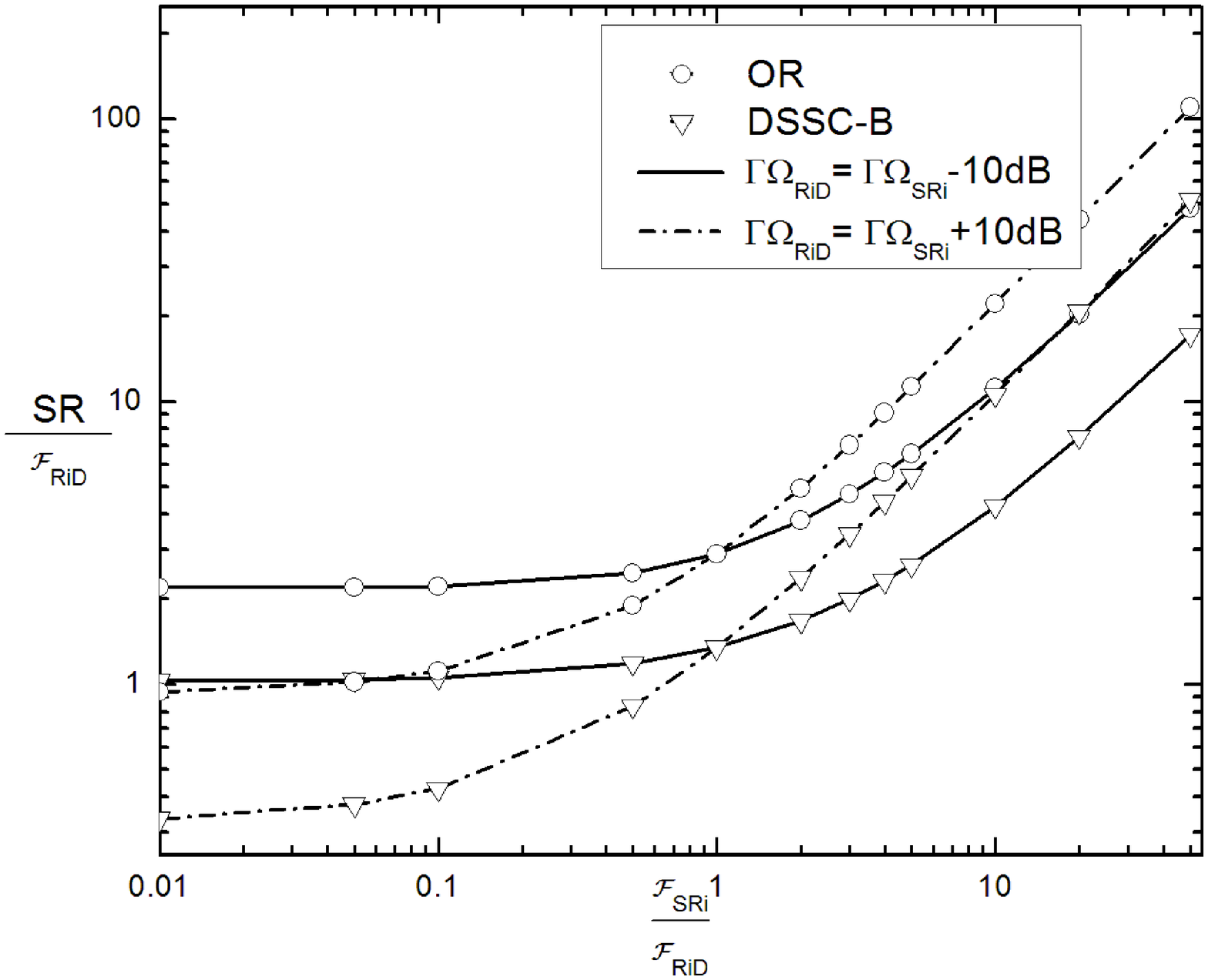}
\caption{Relay switching rates of OR and DSSC}%
\label{SR1}%
\end{center}
\end{figure}
\pagebreak%
\begin{figure}
[ptb]
\begin{center}
\includegraphics[
natheight=8.460500in,
natwidth=10.966700in,
height=4.5792in,
width=5.9283in
]%
{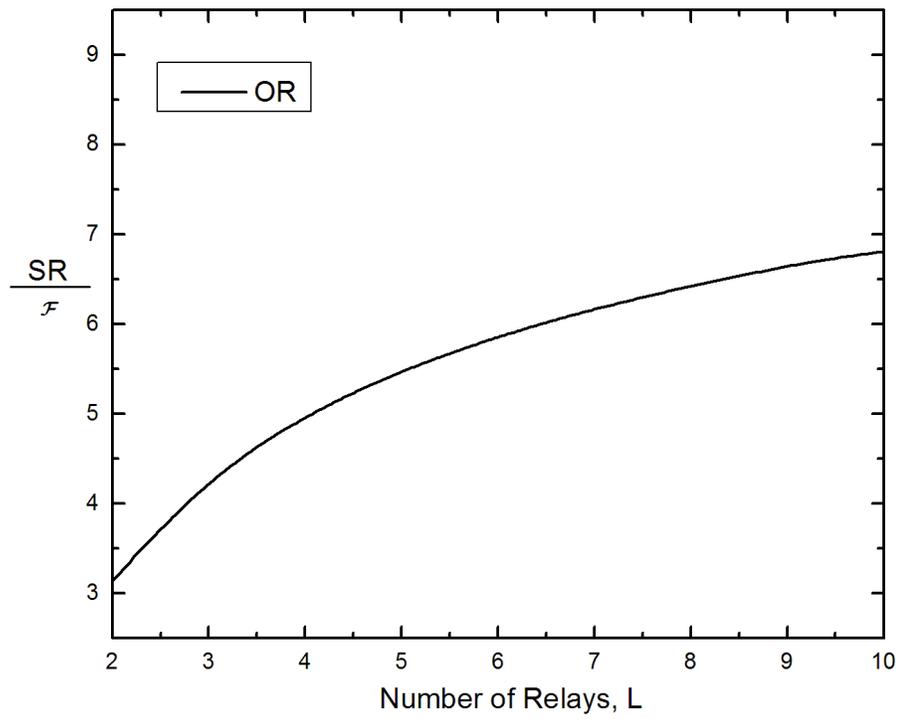}
\caption{Relay switching rate of OR versus the number of available relays}%
\label{SR2}%
\end{center}
\end{figure}
%EndExpansion
\pagebreak%

\begin{figure}
[ptb]
\begin{center}
\includegraphics[
natheight=8.460500in,
natwidth=10.966700in,
height=4.5792in,
width=5.9283in
]%
{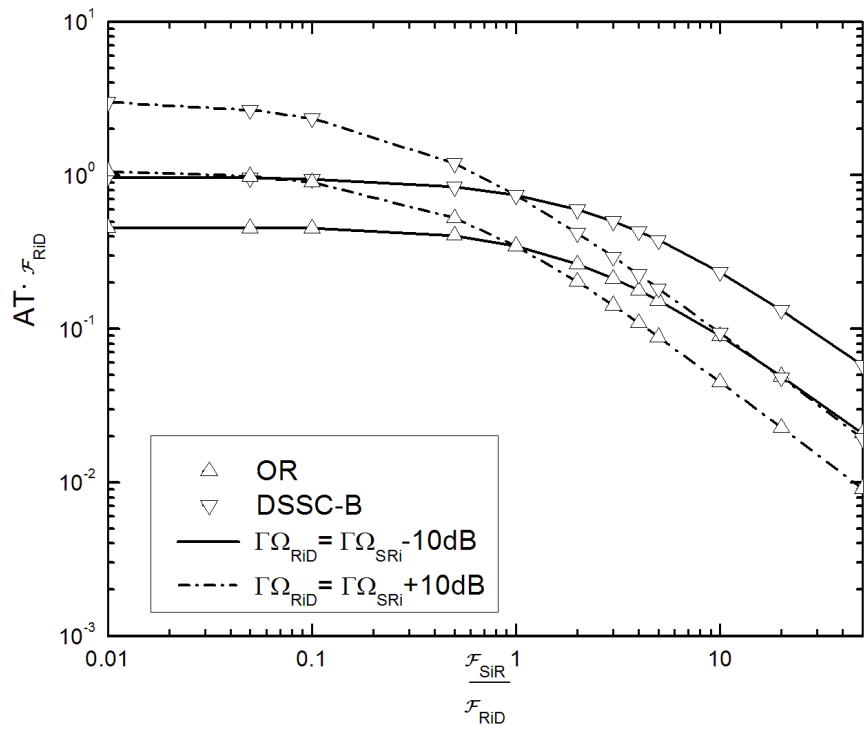}%
\caption{Average activation time of OR and DSSC}%
\label{AAT01}%
\end{center}
\end{figure}
%EndExpansion
\pagebreak%
\begin{figure}
[ptb]
\begin{center}
\includegraphics[
natheight=8.460500in,
natwidth=10.966700in,
height=4.5792in,
width=5.9283in
]%
{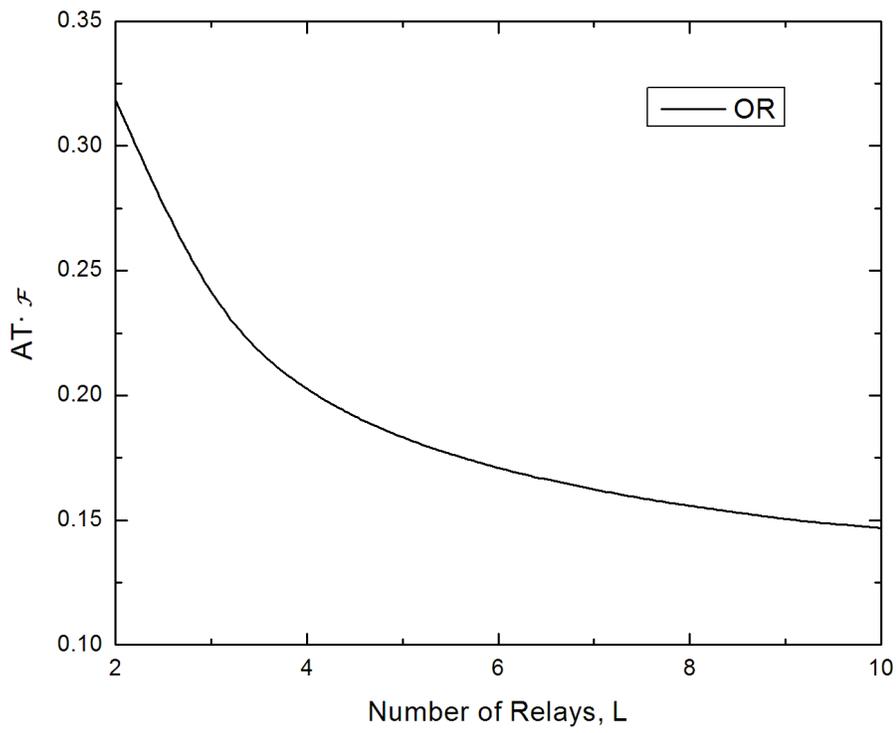}%
\caption{Average activation time of OR versus the number of available relays}%
\label{AAT02}%
\end{center}
\end{figure}
%EndExpansion
\pagebreak\ \pagebreak%
\begin{figure}
[ptb]
\begin{center}
\includegraphics[
natheight=8.460500in,
natwidth=10.966700in,
height=4.5792in,
width=5.9283in
]%
{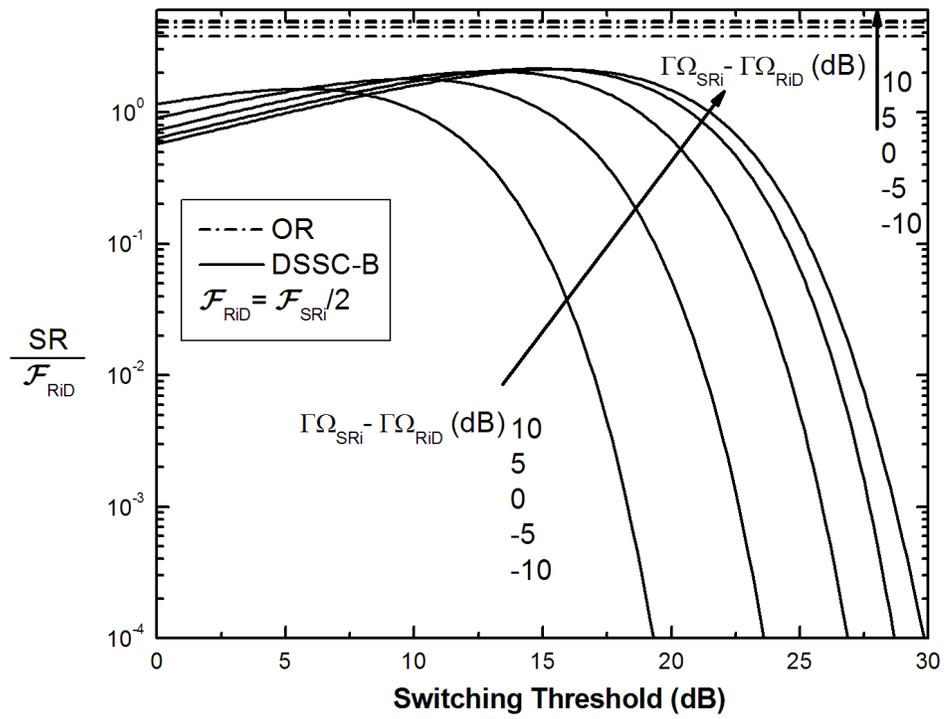}%
\caption{Relay switching rates of OR and DSSC versus the normalized switching
threshold}%
\label{SR3}%
\end{center}
\end{figure}
%EndExpansion

\end{subequations}
\end{document}